\documentclass[10pt,twocolumn,twoside]{IEEEtran}
\usepackage{amssymb,epsfig,xcolor,cite,amsmath,amsfonts,mathrsfs,algorithm}
\usepackage{epstopdf}
\usepackage{datetime,fancyhdr}
\usepackage{algpseudocode}
\usepackage{arydshln}
\usepackage{multirow}
\usepackage{amsthm} 
\allowdisplaybreaks

\usepackage{times} 
\usepackage{epsfig}
\usepackage{graphicx}
\usepackage{latexsym}
\usepackage{subfigure}

\newtheorem{lemma}{Lemma}[section]
\newtheorem{theorem}{Theorem}[section]

\newtheorem{remark}{Remark}[section]

\newtheorem{proposition}{Proposition}[section]
\newtheorem{assumption}{Assumption}[section]
\definecolor{purple}{RGB}{128,0,128} 
\ifCLASSINFOpdf
\else
\fi

\begin{document}

\title{Projection-Free Bandit Online Optimization for Multi-Agent Systems with Dynamic Regret }

\author{Xia Jiang, 
    Lu Liu,~\IEEEmembership{Senior Member,~IEEE,} and Gang Feng,~\IEEEmembership{Fellow,~IEEE}
\thanks{This work was supported by the Research Grants Council of Hong Kong under Grant CityU-11208223. This paper was not presented at any conference. Corresponding author is Lu Liu.}
\thanks{X. Jiang is with the School of Electrical and Electronic Engineering, Nanyang Technological University, Singapore 639798, Singapore (email: xia.jiang@ntu.edu.sg).}
\thanks{L. Liu and G. Feng are with the Department of Mechanical Engineering, City University of Hong Kong, Hong Kong, China (email: luliu45@cityu.edu.hk, megfeng@cityu.edu.hk).}
}

\maketitle

\begin{abstract}
 This paper investigates distributed online optimization for multi-agent dynamical systems with constrained inputs and time-varying cost functions. While online convex optimization offers a principal framework for sequential decision-making, existing online learning and optimization algorithms typically require accurate system models, limiting their applicability in practical settings. To overcome this challenge, we propose a distributed bandit online feedback optimization algorithm that relies solely on real-time input–output data. The algorithm employs a smoothing zeroth-order one-point estimator to construct local gradient approximations directly from cost evaluations. Additionally, to enforce input constraints effectively, we integrate a projection-free conditional gradient update, making the algorithm well-suited for online and large-scale settings. Furthermore, we establish a sublinear dynamic regret bound that depends on a temporal variation measure of system non-stationarity. Finally, numerical simulations demonstrate the effectiveness of the proposed algorithm.
\end{abstract}

\begin{IEEEkeywords}
 bandit feedback, projection-free method, online optimization,  distributed optimization, dynamic regret
\end{IEEEkeywords}

%
\IEEEpeerreviewmaketitle

\section{Introduction}
\par The framework of online convex optimization (OCO) serves as a fundamental tool for modeling sequential decision processes in uncertain environments, offering systematic tools for handling unknown and time-varying cost functions. In the classical OCO setting, a learner makes decisions in a step-by-step manner while facing time-varying and initially unknown cost functions. The quality of these decisions is typically assessed through regret, which measures the performance gap between the learner’s accumulated loss and that of the best fixed decision chosen in hindsight. This paradigm is of significant importance in many practical control applications such as robot formation control \cite{robot_form}, target tracking \cite{onl_mirror_des} and energy flow management \cite{Yuan_online_grids}. Moreover, in many large-scale networked systems, such as smart grids \cite{YI2016259} and intelligent transport \cite{ilte_transp}, the information of the concerned systems is geographically distributed across separated physical nodes, necessitating scalable distributed online optimization (DOO) algorithms. In such multi-agent networks, each node operates as an autonomous decision-making agent that updates its actions based on local information and communicates with its neighbors to collectively optimize the global performance under time-varying objectives and environmental disturbances.
\par Motivated by the need for real-time decision-making in networked systems, considerable efforts have been devoted to developing DOO algorithms, leading to the distributed online dual averaging algorithm \cite{online_dual_ave,di_online_dual_ave}, distributed online primal-dual algorithms \cite{Yi_tac_25,LU2023111203,li_online_cst}, and distributed online gradient tracking algorithms \cite{dis_online_aggre_Va} within the OCO framework. Despite this significant progress, most existing DOO methods overlook the dynamical constraints inherent in physical systems. Consequently, a fundamental gap arises when bridging distributed online learning theory with the real-time operation of complex cyber-physical systems. While some studies have successfully incorporated OCO methodologies into the control of linear dynamical systems, addressing challenges such as external disturbance \cite{pmlr-v139-agarwal21b}, inequality constraints \cite{YU2024111407}, and output feedback \cite{NEURIPS2022_30dfe47a,CARNEVALE2024}, these works critically rely on the availability of accurate model knowledge. In practical scenarios, however, obtaining precise models of the concerned systems is often infeasible due to their inherent nonlinearities, unmodeled dynamics, and environmental variations. 
\par To reduce the dependence on accurate system models, data-driven control has emerged as a promising model-free alternative. Recent works have proposed resilient data-driven model predictive controllers that are robust to measurement noise \cite{data_mpc} and denial-of-service attacks \cite{Liu2022data}, building on behavioral systems theory and historical input–output data. Nevertheless, such approaches typically require persistently exciting input–output data and are restricted to linear time-invariant (LTI) systems. These methods generally assume constant cost functions, treating temporal variations merely as disturbances rather than as part of the optimization objective. In parallel, feedback-based sensitivity learning methods \cite{qin_feed, ES_alg, YANG2023110923}, such as extremum seeking, estimate gradients through perturbations but focus mainly on steady-state optimality. While classical extremum seeking and recent stochastic derivative-free methods \cite{ES_alg, YANG2023110923} have explored stochastic perturbations and averaging to handle time-varying objectives, they primarily offer asymptotic convergence for static problems. Their transient learning performance under non-stationary environments, typically evaluated via dynamic regret, has received very limited attention. Furthermore, extending these methods to general nonlinear systems remains difficult.
\par These observations reveal a remaining gap in the development of effective algorithms capable of achieving online optimization for nonlinear dynamical systems with time-varying cost functions by using only input-output data while with no knowledge of the system model. Recent advancements \cite{review_zerograd,TANG2023110741,mod_fre_feed} in data-driven optimization with bandit feedback provide a useful foundation, showing that zeroth-order gradient-free methods can estimate descent directions from cost evaluations by exploiting system outputs as surrogates for analytical gradients. However, most approaches are tailored for centralized settings and restricted to steady-state optimization with constant cost functions. While some recent efforts have explored bandit feedback for distributed online settings, such as in aggregative games \cite{dis_non_resi_bandit} and event-triggered optimization under constraints \cite{dis_bandit_xinli}, applying these approaches broadly to DOO still faces theoretical and computational challenges. Theoretically, relying solely on limited local observations induces gradient estimation variance \cite{JIANG2026112999}, which couples tightly with network consensus drift in time-varying environments. Additionally, standard projection-based methods suffer from computational burdens under complex constraints, making them impractical for real-time systems. 
\par To bridge this gap, this paper proposes a novel distributed bandit online learning framework for multi-agent systems with input constraints. Specifically, we develop a projection-free distributed feedback optimization algorithm that integrates a smoothing one-point gradient estimator with a network-level conditional gradient update. This framework addresses the coupled challenges of unknown system dynamics and non-stationary environments. 
The main contributions of this work are summarized as follows.
 \begin{itemize}
    \item We establish a model-free distributed online optimization framework for multi-agent systems with unknown nonlinear dynamics. Unlike standard OCO with explicit analytical models, our formulation directly exploits system-induced objective functions via real-time input-output measurements, thereby bridging distributed online learning theory with the real-time operation of complex cyber-physical systems without requiring explicit model knowledge.
     \item We develop the distributed projection-free feedback optimization algorithm (DPFOA) that overcomes the coupled challenges of zeroth-order estimation variance and network consensus drift. Rather than a heuristic combination of existing methods, our approach embeds a smoothing one-point gradient estimator into a refined conditional gradient step. This structural design simultaneously bypasses computationally expensive projection operations and effectively suppresses the propagation of gradient approximation errors across the network. 
      \item We establish a rigorous sublinear dynamic regret bound that explicitly captures the impact of temporal variation in system non-stationarity. The theoretical breakthrough lies in decoupling and bounding the intertwined effects of bandit estimation errors, gradient tracking consensus mismatch, and dynamic temporal variations. This analysis provides the theoretical performance guarantees for bandit online control and explicitly quantifies how underlying system non-stationarity scales the regret bound in distributed zeroth-order settings.
 \end{itemize}
 \par The remainder of the paper is organized as follows. 
Section \ref{solver_design} introduces the problem formulation and presents the distributed online feedback optimization algorithm. Theoretical guarantees on its regret performance are provided in Section \ref{proof_sec}. Numerical experiments demonstrating the effectiveness of the proposed approach are presented in Section \ref{simulation}, followed by concluding remarks in Section \ref{conclusion}.  
\par \textbf{Mathematical Notations:}
Throughout this work, let $\mathbb{R}$ denote the set of real numbers and $\mathbb{R}^n$ the $n$-dimensional Euclidean space. For a positive integer $n$, we use $[n] := {1, \dots, n}$. The transpose of a vector $v$ is denoted by $v^\top$, and $\langle \cdot, \cdot \rangle$ represents the standard inner product. The symbol $\nabla f$ denotes the gradient of a differentiable function $f$. 
\par The communication structure among agents is described by an undirected graph $\mathcal{G}=(\mathcal{V},\mathcal{E},A)$, where $\mathcal{V}=\{1,\cdots,n\}$ is the node set and $\mathcal{E} \subseteq \mathcal{V} \times \mathcal{V}$ is the edge set. Let $A = [A_{ij}]$ denote the associated adjacent matrix, where $A_{ij} > 0$ if agents $i$ and $j$ are connected, and $A_{ij} = 0$ otherwise. For each agent $i$, its neighbor set is defined as $\mathcal{N}_i = \{ j \in \mathcal{V} : (i,j) \in \mathcal{E} \}$.
\section{Problem description and algorithm design}\label{solver_design}

\par This paper investigates an online optimization problem for a multi-agent system comprising $n$ heterogeneous agents. Each agent, say agent $i$, is associated with a stable dynamical sytem that is abstracted by its nonlinear steady-state input-output map $h_i:\mathbb{R}^p\times \mathbb{R}^r\to \mathbb{R}^q$ described as follows:
\begin{align}\label{y_plat_dy}
    y_i=h_i(u_i,d),
\end{align}
where $u_i\in \mathbb{R}^{p}$ is the input, $y_i\in \mathbb{R}^q$ is the output, and $d\in \mathbb{R}^r$ is an unknown but bounded exogenous disturbance. Throughout this work, we do not assume any explicit knowledge of the steady-state map $h_i$.
\par Within the multi-agent network, the agents aim to cooperatively solve a global online optimization problem. Assume that agent $i$ has access to a local, time-varying cost function $\Psi_{i,t}(u_i,y_i)$, which depends on its input $u_{i}$ and the corresponding output $y_{i}$. At every time step $t$, agent $i$ computes its control input $u_{i,t}$ in a decentralized manner and applies it to the underlying dynamical system. Then, the current cost functions $\Psi_{i,t}$ is revealed to the agent. Let $u_t\triangleq [u_{1,t}, \cdots, u_{n,t}]\in\mathbb{R}^{p\times n}$ and $y_t\triangleq [y_{1,t}, \cdots, y_{n,t}]\in \mathbb{R}^{q\times n}$ denote the stacked vectors of local inputs and outputs, respectively.  
The global online optimization problem is formulated as follows:
\begin{align}\label{opti_prob}
    \min_{u,y} \Psi_t(u,y)&=\sum_{i=1}^n \Psi_{i,t}(u_i,y_i)\notag\\
    \text{s.t.} \quad y&=h(u,d), u_i\in \Omega,
\end{align}
where $\Omega \subseteq  \mathbb{R}^{p}$ is the input constraint set and $h=[h_1,\cdots,h_n]$ is the aggregate steady-state map.
By eliminating $y$ via the steady-state relation, problem \eqref{opti_prob} can be equivalently reformulated as an unconstrained optimization problem,
\begin{align}
    \min_u \tilde{\Psi}_t(u)=\sum_{i=1}^n \tilde{\Psi}_{i,t}(u), \quad \text{s.t.} \ u\in \Omega,
\end{align}
where $\tilde{\Psi}_{i,t}(u)=\Psi_{i,t}(u_i,h_i(u_i,d))$ is the reformulated objective function after substituting the steady-state map.  
\par Most existing numerical optimization solvers require an explicit analytical model of the steady-state map $h$ as well as the exact value of the disturbance $d$. However, such information is often difficult or impossible to obtain in the presence of complex system dynamics and unknown external perturbations. To address these challenges, we develop a feedback-based optimization controller that solely utilizes real-time output measurements to iteratively drive system \eqref{y_plat_dy} to an operating point that minimizes the global objective function \eqref{opti_prob}.
\par We make the following assumptions throughout this paper.
\begin{assumption}\label{func_assum}
    The function $\Psi_{i,t}$ is $G$-Lipschitz over the compact set $\Omega$, i.e. for $u_1,u_2\in \Omega$ and $y_1,y_2\in \mathbb{R}^q$,
    $$|\Psi_{i,t}(u_1,y_1)-\Psi_{i,t}(u_2,y_2)|\leq G\|u_1-u_2\|+G\|y_1-y_2\|.$$ 
\end{assumption}
\begin{assumption}\label{psi_tilde_del}
    The function $\tilde{\Psi}_{i,t}$ is $M$-Lipschitz continuous and convex. 
\end{assumption}
\begin{assumption}\label{con_set_ass}
    The constraint set $\Omega$ is convex and compact with diameter $D$, i.e., $\|u_1-u_2\|\leq D$ for any $u_1, u_2\in \Omega$.
\end{assumption}
\begin{assumption}\label{net_assum}
	The multi-agent network is connected and the associated adjacent matrix $A=[A_{ij}]$ is symmetric and doubly stochastic with non-negative elements. 
\end{assumption}
{
\begin{assumption} \label{ass:temporal_variation}
There exists a constant $V > 0$ such that for $i \in \mathcal{V}$ and $t \ge 1$, the time-varying cost function $\Psi_{i,t}$ satisfies
$$|\Psi_{i,t}(u, y) - \Psi_{i,t-1}(u, y)| \le V,$$
for any feasible input $u \in \Omega$ and output $y \in \mathbb{R}^q$.
\end{assumption}
\begin{remark}
Assumptions \ref{func_assum}--\ref{net_assum} are standard in distributed online optimization and multi-agent systems \cite{dsg_nedic,pmlr-v70-zhang17g,Duchi_tac,Extra_shi}. Assumption \ref{func_assum} imposes a Lipschitz condition on the stage cost, while Assumption \ref{psi_tilde_del} requires the reduced objective $\tilde{\Psi}_{i,t}(u_i)$ to be convex and Lipschitz, which ensures that the cost remains well behaved after eliminating the output through the steady-state map. Assumptions \ref{con_set_ass} and \ref{net_assum} further guarantee that the feasible set is compact and that the network can reach average consensus. Assumption \ref{ass:temporal_variation} additionally limits the variation of the cost across consecutive time steps, which is standard in dynamic OCO and prevents abrupt changes in the environment or tracking target. Together, these assumptions provide the regularity and boundedness for the convergence analysis.
\end{remark}

}
\begin{remark}
    To streamline the exposition, this article focuses on nonlinear systems that can be abstracted by steady-state input-output maps. For instance, the dynamical model of agent $i$ can be given as
\begin{align}\label{plant_dynamic}
    x_{i,t+1}=&f_i(x_{i,t},u_{i,t},d),\\
    y_{i,t}=&g_i(x_{i,t},d).
\end{align}
where $u_{i,t}$, $x_{i,t}$, and $y_{i,t}$ denote the input, state and output variables of agent $i$, respectively. The functions $f_i$ and $g_i$ can be different across agents. If the dynamical system \eqref{plant_dynamic} admits a unique, globally exponentially stable steady-state map $x_{i,\mathrm{ss}}(u, d)$, then the corresponding steady-state output satisfies $y_{i}=g_i(x_{i,ss}(u_i,d),d)\triangleq h_i(u_i,d)$, thus recovering the steady-state map representation  \cite{mod_fre_feed}. 
\end{remark}

\begin{algorithm}
\caption{ Distributed projection-free feedback optimization algorithm (DPFOA) for each agent $i\in [n]$ }
	\label{dis_grad_free_alg}
	\begin{algorithmic}[1] 
        \State Initialize: stepsizes $\rho$, $\eta$, $\delta$ and $\sigma_t$, $u_{i,0}, w_{i,0}\in \Omega$, $d_{i,1}=0$, $a_{i,0}=0$, the number of epochs $T$.
        \For {$t=1,\cdots,T$}
        \State  \textit{One-point gradient estimator:}  
        \begin{align}
            g_{i,t}=&\frac{pv_{i,t}}{\delta} \left(\Psi_{i,t}(u_{i,t},y_{i,t})-\Psi_{i,t-1}(u_{i,t-1},y_{i,t-1})\right)\label{git_up}\\
            a_{i,t}=&(1-\rho)a_{i,t-1}+\rho g_{i,t}\label{ait_up}
        \end{align}
        \State \textit{Global gradient estimator:}
        \begin{align}\label{dit_up_step}
            d_{i,t+1}=\sum_{j\in \mathcal{N}_i}A_{ij}d_{j,t}+a_{i,t}
        \end{align}
        
        \State \textit{Projection-free conditional gradient step:}
        \begin{align}
            F_{i,t}(w)=&\eta\langle d_{i,t},w \rangle+\|w\|^2 \label{Fit_up}\\
            s_{i,t}=&\text{argmin}_{s\in \Omega}\langle \nabla F_{i,t}(w_{i,t}),s\rangle  \label{sit_up}\\
            w_{i,t+1}=&w_{i,t}+\sigma_t(s_{i,t}-w_{i,t}) \label{wit_up}
        \end{align}
        \State \textit{Input exploration step:} 
        \begin{align}\label{input_ex_up}
            u_{i,t+1}=w_{i,t+1}+\delta v_{i,t+1}
        \end{align}
        \EndFor
	\end{algorithmic}
\end{algorithm}

\par To solve the online multi-agent optimization problem \eqref{opti_prob}, we design a distributed projection-free feedback optimization algorithm (DPFOA) for each agent $i$ summarized in Algorithm \ref{dis_grad_free_alg}. 
 In DPFOA, we adopt a zeroth-order one-point gradient estimator to approximate local gradients, eliminating the need for explicit gradient computations. This gradient estimator relies solely on  real-time output measurements and thus does not require explicit knowledge of the underlying system model. Moreover, a smoothing step in \eqref{ait_up} is introduced to improve the accuracy of the local gradient estimates. Based on these local gradient approximations, each agent constructs a global gradient estimator by using a weighted combination of its historical local gradients and the information received from its neighbors. To handle the input constraints, we apply a projection-free conditional gradient update with the global gradient estimator $d_{i,t}$. The next control input $u_{i,t+1}$ is then generated by adding an exploration perturbation $\delta v_{i,t+1}$ to the candidate solution $w_{i,t+1}$, and the resulting signal is applied to the system \eqref{plant_dynamic}. 
\par The presence of the constraint set $\Omega$ makes standard Gaussian random vectors $v_{i,t}$ unsuitable for our setting, as the unbounded support of the Gaussian distribution may drive the perturbed input $u_{i,t}=w_{i,t}+\delta v_{i,t}$ outside $\Omega$. To avoid this issue, we sample $v_{i,t}$ uniformly from the unit sphere $\mathbb{S}_{p-1}\triangleq \{v_{i,t}\in \mathbb{R}^p:\|v_{i,t}\|=1\}$, which has been adopted in zeroth-order optimization methods \cite{two_feed,Chen2020FrankWolfe,mod_fre_feed}. In this setting, the smoothed approximation $\tilde{\Psi}_{i,t,\delta}(w)$ for the objective function $\tilde{\Psi}_{i,t}(w)$ is defined as
\begin{align}
    \tilde{\Psi}_{i,t,\delta}(w)=\mathbb{E}_{v' \sim U(\mathbb{B}_p)}[\tilde{\Psi}_{i,t}(w+\delta v')]
\end{align}
where $w\in \mathbb{R}^p$, and $U(\mathbb{B}_p)$ is the uniform distribution over the closed unit ball $\mathbb{B}_p=\{v'\in \mathbb{R}^p|\|v'\|\leq 1\}$ in $\mathbb{R}^p$. The properties of the smoothed approximation $\tilde{\Psi}_{\delta}(w)$ are summarized as follows.
\begin{lemma}\cite{hazan2016introduction}\label{f_del_dif_lem}
     Suppose that $\tilde{\Psi}: \mathbb{R}^p \rightarrow \mathbb{R}$ is $M$-Lipschitz continuous. Then, for any $w \in \mathbb{R}^p$ and $\delta>0$, the smoothed function $\tilde{\Psi}_\delta(w)$ satisfies
\begin{align}
\mathbb{E}_{v \in U\left(\mathbb{S}_{p-1}\right)}\left[\frac{p}{\delta} \tilde{\Psi}(w+\delta v) v\right] & =\nabla \tilde{\Psi}_\delta(w), \label{psi_exp}\\
\left|\tilde{\Psi}_\delta(w)-\tilde{\Psi}(w)\right| & \leq M \delta. \label{del_psi_dif}
\end{align}
\end{lemma}
\begin{remark}
   The local one-point gradient estimator in Algorithm \ref{dis_grad_free_alg} is inspired by existing works \cite{YUAN_OG,hazan2016introduction,pmlr-v119-wan20b,ZHANG_onepoint}. However, a key distinction lies in the use of the approximate objective function value $\Psi_t (u_t, y_{t})$, which relies on real-time output measurements, rather than the exact steady-state objective function $\tilde{\Psi}_t(u_t)=\Psi_t(u_t, h(u_t, d))$. This difference arises from the lack of model information and the presence of system dynamics, which prevent direct access to steady-state outputs at each time step.
\end{remark}

\begin{remark}
To ensure feasibility of the perturbed input \eqref{input_ex_up}, 
we can adopt a standard technique in bandit optimization based on a shrunk feasible set.
Specifically, define
\[
\Omega_\delta := \{ w \in \Omega : \mathrm{dist}(w, \partial \Omega) \ge \delta \}.
\]
By restricting $w_{i,t}$ to $\Omega_\delta$ and noting that $\|v_{i,t}\|\le 1$, 
it follows that $u_{i,t} \in \Omega$, ensuring that all function evaluations are well-defined.
For clarity of presentation, the algorithm is written over the original set $\Omega$. 
The above restriction can be equivalently imposed without affecting the algorithmic structure or the theoretical results. This technique is widely used in zeroth-order and bandit optimization to handle boundary issues.
\end{remark}

\section{Performance analysis}\label{proof_sec}
In this section, we analyze the regret performance of the proposed DPFOA.
To this end, we introduce the following dynamic regret metric \cite{dis_online_aggre_Va}, which measures the accumulated difference between the cost incurred by the proposed algorithm and the optimal cost at each time step:
\begin{align}
    \mathcal{R}_{i,T}=\sum_{j=1}^n\sum_{t=1}^T \left(\mathbb{E}[\tilde{\Psi}_{j,t}(u_{i,t})]-\tilde{\Psi}_{j,t}(u_t^*)\right),
\end{align}
where $u_t^*=\text{argmin}_{u\in \Omega} \tilde{\Psi}_t(u)$ denotes the optimal solution at time $t$. The regret $\mathcal{R}_{i,T}$ is defined from the perspective of agent $i$, using its local decisions $\{u_{i,t}\}_{t=1}^T$, while performance is evaluated with respect to the aggregate cost $\sum_{j=1}^n\tilde{\Psi}_{j,t}(\cdot)$. Thus, $\mathcal{R}_{i,T}$ measures the gap between the network-wide cost incurred by agent $i$’s decisions and that of the optimal centralized benchmark. We further define a  temporal variation metric capturing non-stationarity as $$D_T=\sum_{t=1}^T\|\nabla \tilde{\Psi}_{i,t,\delta}(w_{i,t})-\nabla \tilde{\Psi}_{i,t-1,\delta}(w_{i,t-1})\|^2,$$ which follows the standard formulation adopted in \cite{dynamic_onlie_fw}.

 \par We first establish several technical lemmas. The first lemma provides uniform bounds on the local gradient estimators.

\par \begin{lemma}
    Suppose Assumptions \ref{func_assum}-\ref{con_set_ass} hold. The local gradient estimators $g_{i,t}$ and $a_{i,t}$ remain uniformly bounded. In particular, there exists a constant $\beta>0$ such that $\|g_{i,t}\|\leq \beta$ and $\|a_{i,t}\|\leq \beta$ for all $i\in \mathcal{V}$ and $t$.
\end{lemma}
\begin{proof}
    By the definition of the estimator $g_{i,t}$, one has
    {
    \begin{align*}
        \left\|g_{i, t}\right\| \leq& \frac{p}{\delta}\left\|v_{i, t}\right\| \cdot\left|\Psi_{i, t}\left(u_{i, t}, y_{i, t}\right)-\Psi_{i, t-1}\left(u_{i, t-1}, y_{i, t-1}\right)\right|\\
        \leq &\frac{p}{\delta}\left\|v_{i, t}\right\| \cdot\Big(\underbrace{\left|\Psi_{i, t}\left(u_{i, t}, y_{i, t}\right)-\Psi_{i, t}(u_{i, t-1}, y_{i, t-1})\right|}_{\text{spatial variation}}\\
        &+\underbrace{\left|\Psi_{i, t}(u_{i, t-1}, y_{i, t-1})-\Psi_{i, t-1}(u_{i, t-1}, y_{i, t-1})\right|}_{\text{temporal variation}}\Big)
    \end{align*}}
     Since $v_{i, t}$ is sampled from the unit sphere, it holds that $\left\|v_{i, t}\right\|=1$. We first consider the spatial variation. By Assumptions \ref{func_assum} and \ref{con_set_ass}, we have
\begin{align*}
    \left|\Psi_{i, t}(u_{i, t}, y_{i, t})\!-\!\Psi_{i, t}(u_{i, t-1}, y_{i, t-1})\right| \leq GD\!+G\|y_{i, t}\!-y_{i, t-1}\|. 
\end{align*}
For the output $y_{i, t}=h_i(u_{i, t}, d)$, note that $u_{i, t} \in \Omega$ and the disturbance $d$ is bounded. Hence, the pair $(u_{i, t},d)$ lies in a compact set. Under the standard assumption that the steady-state map $h_i(u_i, d)$ is continuous, it follows that $y_{i, t}$ is uniformly bounded. Therefore, there exists a constant $C_y>0$ such that $\|y_{i, t}\| \leq C_y$ for all $i, t$, which implies $\|y_{i, t}-y_{i, t-1}\| \leq 2 C_y$. We obtain
\begin{align*}
    \left|\Psi_{i, t}(u_{i, t}, y_{i, t})\!-\!\Psi_{i, t}(u_{i, t-1}, y_{i, t-1})\right| \leq G(D+2C_y). 
\end{align*}
{
Then, for the temporal variation term, we know $\left|\Psi_{i, t}(u_{i, t-1}, y_{i, t-1})-\Psi_{i, t-1}(u_{i, t-1}, y_{i, t-1})\right|\leq V$ by Assumption \ref{ass:temporal_variation}. 
Combining the above bounds, we obtain
\begin{align*}
    \left\|g_{i, t}\right\| \leq \frac{p}{\delta}(G(D+2C_y)+V)=: \beta.
\end{align*}}
\par  Next, we show the boundedness of $a_{i,t}$. Since $a_{i,0}=0$ and $a_{i,t}$ is a linear combinition of $a_{i,t-1}$ and $g_{i,t}$ in \eqref{ait_up}, we can prove the boundness of $a_{i,t}$ by induction, i.e. $\|a_{i,t}\|\leq \beta$. Then for all $i\in \mathcal{V}$ and all time $t$, it holds that $\|g_{i,t}\|\leq \beta$ and $\|a_{i,t}\|\leq \beta$. This completes the proof.
\end{proof}
\par The next two lemmas characterize key properties of the global gradient estimator $d_{i,t}$, whose proofs follow directly from Lemma 3 and Lemma 6 in \cite{pmlr-v70-zhang17g}, respectively.
\begin{lemma}\label{consens_vari_dif_lem}
    Suppose Assumption \ref{net_assum} holds. If $\|a_{i,t}\|\leq \beta$, then for any $i\in \mathcal{V}$, we have
    \begin{align}
        \|d_{i,t+1}-d_{i,t}\|\leq \alpha \beta,
    \end{align}
    where $\alpha=\frac{1+\sigma_2(A)}{1-\sigma_2(A)} \sqrt{n}+1$, and $\sigma_2(A)$ denotes the second largest eigenvalue of matrix $A$. 
\end{lemma}
\begin{lemma}\label{average_cons_lem}
Suppose Assumption \ref{net_assum} holds. If $\|a_{i,t}\|\leq \beta$, then for any $i\in \mathcal{V}$, we have
\begin{align}
    \|d_{i,t}-\bar{d}_t\|\leq \alpha' \beta,
\end{align}
where $\bar{d}_t=\frac{1}{n}\sum_{i=1}^n d_{i,t}$ and $\alpha'=\frac{\sqrt{n}}{1-\sigma_2(A)}$.
\end{lemma}
\par The following lemma establishes the Lipschitz-continuity of the projection operator, whose proof can be found in Lemma 5 of \cite{Duchi_tac}.
\begin{lemma}\label{pro_lem}
    Let $\Pi_{\Omega}(u,\eta)=\text{argmin}_{w\in \Omega} \eta u^\top w+\|w\|^2$. Then the operator $\Pi_{\Omega}(\cdot,\eta)$ is Lipschitz continuous in its first argument, i.e.,
    $$\|\Pi_{\Omega}(w,\eta)-\Pi_{\Omega}(v,\eta)\|\leq \eta \|u-v\|.$$
\end{lemma}
\par With the above results, we now discuss the consensus error $\|w_{i,t}-\bar{w}_t\|$ for each agent $i$, where
 \begin{align}\label{bar_W_def}
       \bar{w}_t=\text{argmin}_{w\in \Omega} \bar{F}_t(w) \ \text{with} \ 
       \bar{F}_t(w)=\eta \bar{d}_t^\top w+\|w\|^2.
 \end{align}
 Let $\hat{w}_{i,t}=\text{argmin}_{w\in \Omega} \eta d_{i,t}^\top w+\|w\|^2$. Applying Lemma \ref{pro_lem} yields $\|\hat{w}_{i,t}-\bar{w}_t\|\leq \eta \|d_{i,t}-\bar{d}_t\|\leq \eta \alpha' \beta$, which implies that 
\begin{align}\label{w_dif_1}
    \|w_{i,t}-\bar{w}_t\|\leq& \|w_{i,t}-\hat{w}_{i,t}\|+\|\hat{w}_{i,t}-\bar{w}_t\|\notag\\
    \leq&\sqrt{F_{i,t}(w_{i,t})-F_{i,t}(\hat{w}_{i,t})}+\eta \alpha' \beta,
\end{align}
where the second inequality holds since $F_{i,t}$ is $2$-strongly convex\footnote{A differentiable function $F$ is $2$-strongly convex if for all $x,y$, $F(y)\geq F(x)+\nabla F(x)^\top (y-x)+\|y-x\|^2$.} and $\nabla F_{i,t}(\hat{w}_{i,t})=0$.
Define 
\begin{align*}
    h_{i,t}=F_{i,t}(w_{i,t})-F_{i,t}(\hat{w}_{i,t}).
\end{align*}
We establish the upper bound of $h_{i,t}$ in the following lemma.
\begin{lemma}\label{h_it_lem}
    Suppose Assumptions \ref{func_assum}, \ref{con_set_ass} and \ref{net_assum} hold. If the stepsizes $\eta$ and $\sigma_t=\frac{1}{\sqrt{t}}$ are chosen such that $\eta \alpha \beta \sqrt{h_{i,t+1}}\leq \sigma_{t}^2 D^2$, then, for any $i\in \mathcal{V}$ and any $t=1,\cdots,T$, the following bound holds 
\begin{align}\label{hit_bound}
    h_{i,t}\leq 4D^2 \sigma_t.
\end{align}
where $D$ is the diameter of the constrained set $\Omega$.
\end{lemma}
\begin{proof}
    Following a similar analysis to Lemma 2 of \cite{pmlr-v70-zhang17g}, we obtain
    \begin{align}
    h_{i,t+1} \leq & \left(1-\sigma_{t}\right) h_{i,t}+\sigma_{t}^2 D^2 \notag\\
    & +\eta\left\|d_{i,t+1}-d_{i,t}\right\| \sqrt{h_{i,t+1}}.
    \end{align}
    In addition, by Lemma \ref{consens_vari_dif_lem}, we have known that $\|d_{i,t+1}-d_{i,t}\|\leq \alpha \beta$. Under the chosen stepsizes choice ensuring $\eta\alpha \beta \sqrt{h_{i,t+1}} \leq \sigma_{t}^2 D^2$, the recursion simplifies to
\begin{align*}
    h_{i,t+1} \leq\left(1-\sigma_{t}\right) h_{i,t}+2 D^2 \sigma_{t}^2.
\end{align*}
 Using this recursion, we can prove the bound \eqref{hit_bound} by induction. For the base case $t=1$, by definition, we have
\begin{align*}
h_{i,1} & =F_{i,1}\left(w_{i,1}\right)-F_{i,1}\left(\hat{w}_{i,t}\right) \\
& =\left\|w_{i,1}\right\|^2-\left\|\hat{w}_{i,t}\right\|^2 \\
& \leq 4 D^2 \sigma_{1}.
\end{align*}
\par Next, assume that the bound holds for $t$. Using the recursion, we show that it also holds for $t+1$ :
$$
\begin{aligned}
h_{i,t+1} & \leq\left(1-\sigma_{t}\right) h_{i,t}+2 D^2 \sigma_{t}^2 \\
& \leq 4 D^2 \sigma_{t}\left(1-\sigma_{t}\right)+2 D^2 \sigma_{t}^2 \\
& =4 D^2 \sigma_{t}\left(1-\sigma_{t}+\frac{\sigma_{t}}{2}\right) \\
& =4 D^2 \sigma_{t}\left(1-\frac{\sigma_{t}}{2}\right) \\
& \leq 4 D^2 \sigma_{t+1},
\end{aligned}
$$
where the final inequality follows from the definition of $\sigma_{t}$. This completes the induction and thus the proof.
\end{proof}
\begin{remark}
To satisfy the condition in Lemma \ref{h_it_lem}, one can choose $\eta=\frac{(1-\sigma_2(A))D}{2(\sqrt{n}+1+(\sqrt{n}-1)\sigma_2(A))\beta T^{3/4}}$. With some straightforward analysis, it can be verified that this choice of $\eta$ indeed fulfills the condition required in Lemma \ref{h_it_lem}. A detailed derivation can be found in Appendix E of \cite{pmlr-v70-zhang17g}.
\end{remark}
\par Using Lemma \ref{h_it_lem}, we can further bound the consensus error in \eqref{w_dif_1}, as stated in the following proposition.
\begin{proposition}
    Suppose Assumptions \ref{func_assum}, \ref{con_set_ass} and \ref{net_assum} hold. The cumulative deviation between $w_{i,t}$ and $\bar{w}_t$ over horizon $T$ satisfies \begin{align}\label{consensus_inter_result}
    \sum_{t=1}^T \|w_{i,t}-\bar{w}_t\| \leq \frac{8}{3}D T^{3/4}+\eta \alpha' \beta  T,
    \end{align}
    where $\alpha'=\frac{\sqrt{n}}{1-\sigma_2(A)}$.
\end{proposition}
\begin{proof}
Using \eqref{w_dif_1} and Lemma \ref{h_it_lem}, we have
    \begin{align*}
    \|w_{i,t}-\bar{w}_t\|\leq &\sqrt{h_{i,t}}+\eta \alpha' \beta\\
    \leq& 2D  \sqrt{\sigma_t}+\eta \alpha' \beta\\
    \leq& 2D  t^{-1/4}+\eta \alpha' \beta.
\end{align*}
By summing from $t=1$ to $T$, we have that for all $i\in \mathcal{V}$, 
\begin{align*}
    \sum_{t=1}^T \|w_{i,t}-\bar{w}_t\| \leq \frac{8}{3}D  T^{3/4}+\eta \alpha' \beta  T.
\end{align*}
This completes the proof.
\end{proof}

 Using the consensus error bound in \eqref{consensus_inter_result}, we now analyze the regret of the proposed DPFOA algorithm. Without loss of generality, we let $u_t^*=w_t^*$. 
For any $i,j\in \mathcal{V}$, we have
\begin{align}\label{Psi_ineq}
    &\sum_{t=1}^T \tilde{\Psi}_{j,t}(u_{i,t})-\sum_{t=1}^T \tilde{\Psi}_{j,t}(u_t^*)\notag\\
    =&\sum_{t=1}^T  \tilde{\Psi}_{j,t}(w_{i,t}+\delta v_{i,t})-\sum_{t=1}^T \tilde{\Psi}_{j,t}(w_t^*)\notag\\
    \leq& \sum_{t=1}^T\left(\tilde{\Psi}_{j,t}(w_{i,t})+M\|\delta v_{i,t}\|\right)-\sum_{t=1}^T\tilde{\Psi}_{j,t}(w_t^*)\notag\\
    \leq&\sum_{t=1}^T \tilde{\Psi}_{j,t}(w_{i,t})-\sum_{t=1}^T \tilde{\Psi}_{j,t}(w_t^*)+\delta M T\notag\\
    \leq&\sum_{t=1}^T \left(\tilde{\Psi}_{j,t,\delta}(w_{i,t})\!+\!\delta M \!\right)\!-\!\sum_{t=1}^T\left(\tilde{\Psi}_{j,t,\delta}(w_t^*)\!-\!\delta M \!\right)\!+\!\delta M T\notag\\
    \leq& \sum_{t=1}^T \left(\tilde{\Psi}_{j,t,\delta}(w_{i,t})-\tilde{\Psi}_{j,t,\delta}(w_t^*)\right)+3 \delta M T,
\end{align}
where the first inequality follows from the $M$-Lipschitz continuity of $\tilde{\Psi}_{i,t}$ in Assumption \ref{psi_tilde_del} and the third inequality follows from Lemma \ref{f_del_dif_lem}.
\par 
Using Assumption \ref{psi_tilde_del} together with the inequality \eqref{consensus_inter_result}, we now analyze the remaining term $\tilde{\Psi}_{j,t,\delta}(w_{i,t})-\tilde{\Psi}_{j,t,\delta}(w_t^*)$, which appears on the right-hand side of \eqref{Psi_ineq}.
Specifically, for any $i,j\in \mathcal{V}$, we have
\begin{align}\label{Psi_dif_2}
    &\sum_{t=1}^T\tilde{\Psi}_{j,t,\delta}(w_{i,t})-\tilde{\Psi}_{j,t,\delta}(w_t^*) \notag\\
    \leq&\sum_{t=1}^T \tilde{\Psi}_{j,t,
    \delta}(w_{j,t})-\tilde{\Psi}_{j,t,\delta}(w_t^*)+\sum_{t=1}^T M\|w_{j,t}-\bar{w}_t\|\notag\\
    &+\sum_{t=1}^T M\|w_{i,t}-\bar{w}_t\|\notag\\
    \leq& \sum_{t=1}^T \nabla \tilde{\Psi}_{j,t,
    \delta}(w_{j,t})^\top (w_{j,t}\!-w_t^*)\!+\!2 M (\frac{8}{3}D T^{3/4}\!+\!\eta \alpha' \beta  T),
\end{align}
where the first inequality holds by Assumption \ref{psi_tilde_del} and the second inequality follows from \eqref{consensus_inter_result}.
\par Next, we focus on the upper bound for the term \\$\sum_{t=1}^T \nabla \tilde{\Psi}_{j,t,\delta}(w_{j,t})^\top (w_{j,t}-w_t^*)$ appearing on the right-hand side of \eqref{Psi_dif_2}.
For clarity, we define  
\begin{align}\label{D_t_def}
    \mathcal{P}_t= (\nabla \tilde{\Psi}_{j,t,\delta}(w_{j,t})-a_{j,t})^\top (w_{j,t}-w_t^*).
\end{align}
Using the boundedness of $\Omega$, it follows that
\begin{align*}
    \sum_{t=1}^T \mathcal{P}_t\leq & \sum_{t=1}^T \|\nabla \tilde{\Psi}_{j,t,\delta}(w_{j,t})-a_{j,t}\|\|w_{j,t}-w_t^*\|\\
    \leq& D   \sum_{t=1}^T \|\nabla \tilde{\Psi}_{j,t,\delta}(w_{j,t})-a_{j,t}\|.
\end{align*}
\par The next proposition provides an upper bound for the cumulative deviation $\sum_{t=1}^T \|\nabla \tilde{\Psi}_{i,t,\delta}(w_{i,t})-a_{i,t}\|$.
\begin{proposition}\label{sum_norm_sub_lem}
    Suppose Assumptions \ref{func_assum}-\ref{net_assum} hold. If the stepsizes $\rho=\frac{1}{T^c}$ with $c\in (0,1)$ and $\delta=\frac{1}{T^{b}}$ with $b\in (0,1)$, we have
\begin{align*}
    &\sum_{t=1}^T\mathbb{E}\left[\|\nabla \tilde{\Psi}_{i,t,\delta}(w_{i,t})-a_{i,t}\|\right]\\
    &=\mathcal{O}(T^{1+b-\frac{c}{2}})+\mathcal{O}(T^{\frac{1}{2}+b+\frac{c}{2}})+\mathcal{O}(T^{c+\frac{1}{2}})\sqrt{D_T}.
\end{align*}
\end{proposition}
\begin{proof}
    At first, we define an auxiliary gradient estimator $\tilde{g}_{i,t}=\frac{p v_{i,t}}{\delta}\left(\tilde{\Psi}_{i,t}(u_{i,t})-\tilde{\Psi}_{i,t}(u_{i,t-1})\right)$. By \eqref{input_ex_up} in the proposed algorithm and \eqref{psi_exp} in Lemma \ref{f_del_dif_lem}, the estimator $\tilde{g}_{i,t}$ satisfies $\mathbb{E}_{v_{i,[t]}}\left[\tilde{g}_{i,t}\right]=\nabla \tilde{\Psi}_{i,t,\delta}\left(w_{i,t}\right)$. 
\par Using the smoothing update \eqref{ait_up}, the term $\|\nabla \tilde{\Psi}_{i,t,\delta}(w_{i,t})-a_{i,t}\|$ admits the following bound:
\begin{align*}
    &\|\nabla \tilde{\Psi}_{i,t,\delta}(w_{i,t})-a_{i,t}\|\\
    \leq&\|\nabla \tilde{\Psi}_{i,t,\delta}(w_{i,t})-(1-\rho)a_{i,t-1}-\rho g_{i,t}+\rho\tilde{g}_{i,t}-\rho \tilde{g}_{i,t}\|\\
    \leq& \|\nabla \tilde{\Psi}_{i,t,\delta}(w_{i,t})-(1-\rho)a_{i,t-1}\!-\!\rho \tilde{g}_{i,t}\|+\!\|\rho(\tilde{g}_{i,t}\!-g_{i,t})\|.
\end{align*}
Squaring both sides of the above inequality yields
\begin{align}\label{square_nabla_ait_dif}
    &\|\nabla \tilde{\Psi}_{i,t,\delta}(w_{i,t})-a_{i,t}\|^2\notag\\
    \leq& (1+\tau)\|\nabla \tilde{\Psi}_{i,t,\delta}(w_{i,t})-(1-\rho)a_{i,t-1}-\rho \tilde{g}_{i,t}\|^2\notag\\
    &+(1+\frac{1}{\tau})\rho^2\|\tilde{g}_{i,t}-g_{i,t}\|^2\notag\\
     \leq&(1+\tau)\|\nabla \tilde{\Psi}_{i,t,\delta}(w_{i,t})-(1-\rho)a_{i,t-1}-\rho \tilde{g}_{i,t}\|^2\notag\\
     &+(1+\frac{1}{\tau})4\rho^2\beta^2,
\end{align}
where the parameter $\tau={\frac{\rho}{2(1-\rho)}}>0$ is chosen such that $\frac{\tau}{1+\tau}<\rho$.
\par For the term $\|\nabla \tilde{\Psi}_{i,t,\delta}(w_{i,t})-(1-\rho)a_{i,t-1}-\rho \tilde{g}_{i,t}\|^2$, by adding and subtracting the term $(1-\rho)\nabla \tilde{\Psi}_{i,t-1,\delta}(w_{i,t-1})$ within the norm and rearranging terms, we obtain
\begin{align*}
&\| \nabla \tilde{\Psi}_{i,t,\delta}(w_{i,t})-(1-\rho)a_{i,t-1}-\rho \tilde{g}_{i,t}\|^2 \notag\\
=& \left\|\rho\left(\nabla \tilde{\Psi}_{i,t,\delta}(w_{i,t})-\tilde{g}_{i,t}\right)\right\|^2+\left\|\rho'\tilde{D}_{i,t}\right\|^2 \notag\\
& +\left\|\rho'\left(\nabla \tilde{\Psi}_{i,t-1,\delta}(w_{i,t-1})-a_{i,t-1}\right)\right\|^2 \notag\\
& +2 \rho \rho'\left\langle\nabla \tilde{\Psi}_{i,t,\delta}(w_{i,t})-\tilde{g}_{i,t}, \tilde{D}_{i,t}\right\rangle \notag\\
&+2\left(\rho'\right)^2\left\langle\tilde{D}_{i,t}, \nabla \tilde{\Psi}_{i,t-1,\delta}(w_{i,t-1})-a_{i,t-1}\right\rangle \notag\\
&+2\left(\rho'\right)^2\!\left\langle \! \nabla \tilde{\Psi}_{i,t,\delta}(w_{i,t})\!-\!\tilde{g}_{i,t}, \nabla \tilde{\Psi}_{i,t-1,\delta}(w_{i,t-1}\!)\!-\!a_{i,t-1} \!\right\rangle,
\end{align*}
where $\tilde{D}_{i,t}=\nabla \tilde{\Psi}_{i,t,\delta}(w_{i,t})-\nabla \tilde{\Psi}_{i,t-1,\delta}(w_{i,t-1})$ and $\rho'=1-\rho$. Taking conditional expectation $\mathbb{E}_{v_{i,[t]}}[\cdot]$ and using $\mathbb{E}_{v_{i,[t]}}\left[\tilde{g}_{i,t}\right]= \nabla \tilde{\Psi}_{i,t,\delta}\left(w_{i,t}\right)$, we obtain
\begin{align}\label{exp_eq_1}
&\mathbb{E}_{v_{i,[t]}} \left[\| \nabla \tilde{\Psi}_{i,t,\delta}(w_{i,t})-(1-\rho)a_{i,t-1}-\rho \tilde{g}_{i,t}\|^2\right] \notag\\
= & \mathbb{E}_{v_{i,[t]}}\left[\left\|\rho\left(\nabla \tilde{\Psi}_{i,t,\delta}(w_{i,t})-\tilde{g}_{i,t}\right)\right\|^2 \right]+\left\|\rho'\tilde{D}_{i,t}\right\|^2 \notag\\
& +2\left(\rho'\right)^2\left\langle\tilde{D}_{i,t}, \nabla \tilde{\Psi}_{i,t-1,\delta}(w_{i,t-1})-a_{i,t-1}\right\rangle  \notag\\
& +\left\|\rho'\left(\nabla \tilde{\Psi}_{i,t-1,\delta}(w_{i,t-1})-a_{i,t-1}\right)\right\|^2.
\end{align}
By Young's inequality, it holds that 
\begin{align}\label{inn_prod_2}
&\left\langle\tilde{D}_{i,t}, \nabla \tilde{\Psi}_{i,t-1,\delta}(w_{i,t-1})-a_{i,t-1}\right\rangle \notag\\
\leq & \frac{\rho}{2}\left\|\nabla \tilde{\Psi}_{i,t-1,\delta}(w_{i,t-1})-a_{i,t-1}\right\|^2+\frac{1}{2 \rho}\left\|\tilde{D}_{i,t}\right\|^2 .
\end{align}
Substituting \eqref{inn_prod_2} into \eqref{exp_eq_1} and using $(1+\frac{1}{\rho})\rho'\leq \frac{1}{\rho}$, we obtain
\begin{align}\label{exp_final_dif}
&\mathbb{E}_{v_{i,[t]}} \left[\| \nabla \tilde{\Psi}_{i,t,\delta}(w_{i,t})-(1-\rho)a_{i,t-1}-\rho \tilde{g}_{i,t}\|^2\right] \notag\\
\leq & \rho^2 \mathbb{E}_{v_{i,[t]}}\left[\left\|\nabla \tilde{\Psi}_{i,t,\delta}(w_{i,t})-\tilde{g}_{i,t}\right\|^2 \right] \notag\\
& +\frac{\rho'}{\rho}\left\|\nabla \tilde{\Psi}_{i,t,\delta}(w_{i,t})-\nabla \tilde{\Psi}_{i,t-1,\delta}(w_{i,t-1})\right\|^2 \notag\\
& +\rho'\left\|\nabla \tilde{\Psi}_{i,t-1,\delta}(w_{i,t-1})-a_{i,t-1}\right\|^2.
\end{align}
Taking total expectation of \eqref{exp_final_dif} yields
\begin{align}\label{exp_fff_2}
    &\mathbb{E}\left[\|\nabla \tilde{\Psi}_{i,t,\delta}(w_{i,t})-(1-\rho)a_{i,t-1}-\rho \tilde{g}_{i,t}\|^2\right]\notag\\
    \leq &\rho^2 \mathbb{E}\left[\|\nabla \tilde{\Psi}_{i,t,\delta}(w_{i,t})-\tilde{g}_{i,t}\|^2\right]\notag\\
    &+\frac{1-\rho}{\rho}\mathbb{E}\left[\|\nabla \tilde{\Psi}_{i,t,\delta}(w_{i,t})-\nabla \tilde{\Psi}_{i,t-1,\delta}(w_{i,t-1})\|^2\right]\notag\\
    &+(1-\rho)\mathbb{E}\left[\|\nabla \tilde{\Psi}_{i,t-1,\delta}(w_{i,t-1})-a_{i,t-1}\|^2\right].
\end{align}
Furthermore, the variance term satisfies 
\begin{align}\label{var_ineq}
    &\mathbb{E}\left[\|\nabla \tilde{\Psi}_{i,t,\delta}(w_{i,t})-\tilde{g}_{i,t}\|^2\right]\notag\\
    =&\mathbb{E}\left[\|\frac{p v_{i,t}}{\delta}\left(\tilde{\Psi}_{i,t}(u_{i,t})-\tilde{\Psi}_{i,t}(u_{i,t-1})\right)\!-\!\nabla \tilde{\Psi}_{i,t,\delta}(w_{i,t})\|^2\right]\notag\\
    \leq& \mathbb{E}\left[\|\frac{p v_{i,t}}{\delta}\left(\tilde{\Psi}_{i,t}(u_{i,t})-\tilde{\Psi}_{i,t}(u_{i,t-1})\right)\|^2\right]\notag\\
    \leq &\frac{2M^2 p^2}{\delta^2}(\sigma_t^2 D^2+4\delta^2)\notag\\
    \leq &\frac{2M^2 p^2}{\delta^2}( D^2+4\delta^2)=\mathcal{O}(T^{2b}),
\end{align}
where the first inequality holds since the variance is bounded by the second moment, the second inequality follows from Assumptions \ref{psi_tilde_del} and \ref{con_set_ass} and {we use $\delta=T^{-b}$ in the last equality}.
\par Taking expectation of \eqref{square_nabla_ait_dif}, summing from $t=1$ to $T$ and using \eqref{exp_fff_2}, \eqref{var_ineq}, we have
\begin{align*}
    &\sum_{t=1}^T\mathbb{E}\left[\|\nabla \tilde{\Psi}_{i,t,\delta}(w_{i,t})-a_{i,t}\|^2\right]\\
    \leq& (1+\tau)\rho^2 \sum_{t=1}^T\mathbb{E}\left[\|\nabla \tilde{\Psi}_{i,t,\delta}(w_{i,t})-\tilde{g}_{i,t}\|^2\right]\\
    &+\!\frac{(1\!+\!\tau)(1\!-\!\rho)}{\rho}\!\sum_{t=1}^T\mathbb{E}\left[\|\nabla \tilde{\Psi}_{i,t,\delta}(w_{i,t})\!-\!\nabla \tilde{\Psi}_{i,t-1,\delta}(w_{i,t-1})\|^2\right]\\
    &+(1+\tau)(1-\rho)\sum_{t=1}^T\mathbb{E}\left[\|\nabla \tilde{\Psi}_{i,t-1,\delta}(w_{i,t-1})-a_{i,t-1}\|^2\right]\\
    &+{8\rho \beta^2}\\
    \leq& (1+\tau)\rho^2 \mathcal{O}(T^{1+2b})+\frac{(1+\tau)(1-\rho)}{\rho}D_T\\
    &+(1\!+\!\tau)(1\!-\!\rho)\sum_{t=0}^{T-1}\mathbb{E}\left[\|\nabla \tilde{\Psi}_{i,t,\delta}(w_{i,t})-a_{i,t}\|^2\right]+{8\rho \beta^2}, 
\end{align*}
where $D_T=\sum_{t=1}^T\|\nabla \tilde{\Psi}_{i,t,\delta}(w_{i,t})-\nabla \tilde{\Psi}_{i,t-1,\delta}(w_{i,t-1})\|^2$.
\par Rearranging terms and using $\|\nabla \tilde{\Psi}_{i,T,\delta}(w_{i,T})-a_{i,T}\|^2\geq (\rho(1+\tau)-\tau) \|\nabla \tilde{\Psi}_{i,T,\delta}(w_{i,T})-a_{i,T}\|^2$ since $\rho<1$, we get
\begin{align*}
    &(\rho(\tau+1)-\tau)\sum_{t=1}^T\mathbb{E}\left[\|\nabla \tilde{\Psi}_{i,t,\delta}(w_{i,t})-a_{i,t}\|^2\right]\\
    \leq &(1+\tau)\rho^2 \mathcal{O}(T^{1+2b})+\frac{(1+\tau)(1-\rho)}{\rho}D_T\\
    &+(1+\tau)(1-\rho)\|\nabla \tilde{\Psi}_{i,0,\delta}(w_{i,t})-a_{i,0}\|^2+{8\rho \beta^2}\\
    = &(1+\tau)\rho^2 \mathcal{O}(T^{1+2b})+\frac{(1+\tau)(1-\rho)}{\rho}D_T\\
    &+(1+\tau)(1-\rho)\|\nabla \tilde{\Psi}_{i,0,\delta}(w_{i,t})\|^2+{8\rho \beta^2},
\end{align*}
where the last equality is due to $a_{i,0}=0$.
Let $\tau'=1+\tau$. Dividing both sides by $\rho\tau'-\tau>0$ yields
\begin{align}\label{sum_grad_dif_exp}
   &\sum_{t=1}^T\mathbb{E}\left[\|\nabla \tilde{\Psi}_{i,t,\delta}(w_{i,t})-a_{i,t}\|^2\right]\notag\\
    \leq &\frac{\tau'\rho^2}{\rho\tau'-\tau} \mathcal{O}(T^{1+2b})+\frac{\tau'(1-\rho)}{\rho(\rho\tau'-\tau) }D_T \notag\\
    &+\frac{\tau'(1-\rho)}{\rho\tau'-\tau}\|\nabla \tilde{\Psi}_{i,0,\delta}(w_{i,t})\|^2+\frac{{8\rho \beta^2}}{\rho\tau'-\tau}.
\end{align}
Finally, by Cauchy–Schwarz, we have
\begin{align}\label{cs_ineq}
    &\sum_{t=1}^T\mathbb{E}\left[\|\nabla \tilde{\Psi}_{i,t,\delta}(w_{i,t})-a_{i,t}\|\right]\notag\\
    \leq &\sqrt{T\sum_{t=1}^T\mathbb{E}\left[\|\nabla \tilde{\Psi}_{i,t,\delta}(w_{i,t})-a_{i,t}\|^2\right]}.
    \end{align}
    It then follows from \eqref{sum_grad_dif_exp} and \eqref{cs_ineq} that
\begin{align*}
    &\sum_{t=1}^T\mathbb{E}\left[\|\nabla \tilde{\Psi}_{i,t,\delta}(w_{i,t})-a_{i,t}\|\right]\\
    \leq& \sqrt{\tau'}{T^{1+b}} \sqrt{\frac{\rho^2}{\rho\tau'-\tau}}+2\sqrt{2T} \beta  \sqrt{\frac{\rho}{\rho\tau'-\tau}}\\
    &+\sqrt{\tau'}\sqrt{\frac{T}{\rho\tau'-\tau}}\|\nabla \tilde{\Psi}_{i,0,\delta}(w_{i,t})\|\\
    &+\sqrt{\tau'}\sqrt{\frac{1}{\rho(\rho\tau'-\tau)}}\sqrt{TD_T}.
\end{align*}

By using $\rho=\frac{1}{T^c}$ and $\beta=\mathcal{O}(1/\delta)=\mathcal{O}(T^{b})$, it follows that
\begin{align*}
    &\sum_{t=1}^T\mathbb{E}\left[\|\nabla \tilde{\Psi}_{i,t,\delta}(w_{i,t})-a_{i,t}\|\right]\\
    =&\mathcal{O}(T^{1+b-\frac{c}{2}})+\mathcal{O}(T^{\frac{1}{2}+b+\frac{c}{2}})+\mathcal{O}(T^{c+\frac{1}{2}})\sqrt{D_T}.
    \end{align*}
This completes the proof. 
\end{proof}

\par Define $\bar{a}_t=\frac{1}{n}\sum_{i=1}^n a_{i,t}$. Then, using the definition of $\mathcal{P}_t$ in \eqref{D_t_def}, we obtain
\begin{align}\label{inn_prod_sum}
    &\sum_{j=1}^n\sum_{t=1}^T\nabla \tilde{\Psi}_{j,t,\delta}(w_{j,t})^\top (w_{j,t}-w_t^*)\notag\\
    =&\sum_{j=1}^n\sum_{t=1}^T  (\nabla \tilde{\Psi}_{j,t,\delta}(w_{j,t})-a_{j,t})^\top (w_{j,t}-w_t^*)\notag\\
    &+\sum_{j=1}^n \sum_{t=1}^T a_{j,t}^\top (w_{j,t}-\bar{w}_t)+n\sum_{t=1}^T \bar{a}_t^\top (\bar{w}_t-w_t^*)\notag\\
    \leq& \sum_{j=1}^n\sum_{t=1}^T \mathcal{P}_t+n\beta \left(\frac{8}{3}D T^{3/4}+\eta \alpha' \beta  T\right)\notag\\
    &+n\sum_{t=1}^T\bar{a}_t^\top (\bar{w}_t-w_t^*),
\end{align}
where the last inequality follows from the consensus error bound in \eqref{consensus_inter_result} and the boundedness of $a_{i,t}$.
\par To further bound the term $\sum_{t=1}^T\bar{a}_t^\top (\bar{w}_t-w^*)$, we introduce the following lemma.
\begin{lemma}\cite{ShalevShwartz2011}[Lemma 2.3]\label{f_ref_lem}
    Let $\widehat{w}_t^*=$ $\underset{w \in \Omega}{\operatorname{argmin}} \sum_{i=1}^{t-1} f_i(w)+\mathcal{R}(w), \forall t \in[T]$, where $\mathcal{R}(w)$ is a strongly convex regularization function. Then, $\forall w \in \Omega$, it holds that

$$
\begin{aligned}
& \sum_{t=1}^T\left(f_t\left(\widehat{w}_t^*\right)-f_t(w)\right) \\
\leq & \mathcal{R}(w)-\mathcal{R}\left(\widehat{w}_1^*\right)+\sum_{t=1}^T\left(f_t\left(\widehat{w}_t^*\right)-f_t\left(\widehat{w}_{t+1}^*\right)\right).
\end{aligned}
$$
\end{lemma}
Applying Lemma \ref{f_ref_lem}, we establish an upper bound on the last term in \eqref{inn_prod_sum} as follows.
\begin{proposition}\label{at_in_bound_lem} 
Suppose Assumption \ref{con_set_ass} holds. Then, we have
    \begin{align*}
    \sum_{t=1}^T \bar{a}_t^\top(\bar{w}_t-w_t^*)\leq \frac{D^2}{\eta}+\eta T \beta^2.
\end{align*}

\end{proposition}
\begin{proof}
{Define $\hat{w}=\arg\min_{w\in \Omega}\sum_{t=1}^T \bar{a}_t^\top w$. Then, it holds that $\sum_{t=1}^T \bar{a}_t^\top(\bar{w}_t-w_t^*)\leq \sum_{t=1}^T \bar{a}_t^\top (\bar{w}_t-\hat{w})$.
}
By \eqref{dit_up_step} in the proposed algorithm, we have $\bar{d}_{t+1}=\bar{d}_t+\bar{a}_t$. Note that $\bar{w}_{t+1}$ is given by $\bar{w}_{t+1}=\text{argmin}_{w\in \Omega} \eta \bar{d}_{t+1}^\top w+\|w\|^2$. Applying Lemma \ref{f_ref_lem} to the linear loss sequence $\{\bar{a}_t^\top w\}_{t=1}^T$ with the regularizer $\mathcal{R}(w)=\frac{\|w\|^2}{\eta}$, we obtain 
\begin{align*}
    \sum_{t=1}^T \bar{a}_t^\top(\bar{w}_t-w_t^*)\leq& \frac{\|\hat{w}\|^2}{\eta}-0+\sum_{t=1}^T \bar{a}_t^\top(\bar{w}_t-\bar{w}_{t+1})\\
    \leq&\frac{D^2}{\eta}+\sum_{t=1}^T \|\bar{a}_t\|\|\bar{w}_t-\bar{w}_{t+1}\|,
\end{align*}
where the last inequality follows from the boundedness of $\Omega$ and the Cauchy-Schwarz inequality.
\par By \eqref{bar_W_def}, it is straightforward to verify that $\bar{F}_t(w)$ is $2$-strongly convex and $\nabla \bar{F}_{t+1}(\bar{w}_{t+1})=0$, then one has 
\begin{align*}
    \|\bar{w}_t-\bar{w}_{t+1}\|^2\leq&\bar{F}_{t+1}(\bar{w}_t)-\bar{F}_{t+1}(\bar{w}_{t+1})\\
    =&\bar{F}_t(\bar{w}_t)+\eta \bar{a}_t^\top \bar{w}_t-\bar{F}_t(\bar{w}_{t+1})-\eta \bar{a}_t^\top \bar{w}_{t+1}\\
    =&\bar{F}_t(\bar{w}_t)-\bar{F}_t(\bar{w}_{t+1})+\eta \bar{a}_t^\top(\bar{w}_t-\bar{w}_{t+1})\\
    \leq& \eta \|\bar{a}_t\|\|\bar{w}_t-\bar{w}_{t+1}\|,
\end{align*}
which implies $\|\bar{w}_t-\bar{w}_{t+1}\|\leq \eta \|\bar{a}_t\|$.  
Combining this result with the boundedness condition $\|\bar{a}_t\|\leq \beta$, we obtain the desired result.
\end{proof}
Building on the upper bounds established in Propositions \ref{sum_norm_sub_lem} and \ref{at_in_bound_lem}, we now derive the dynamic regret bound for the proposed DPFOA algorithm.
\begin{theorem}\label{reg_them}
    Suppose Assumptions \ref{func_assum}-\ref{net_assum} hold. If the stepsizes are chosen as $\sigma_t=\frac{1}{\sqrt{t}}$,  $\rho=T^{-2/5}$, $\eta=T^{-3/4}$ and $\delta=T^{-1/10}$, then we have
    \begin{align}
        &\sum_{j=1}^n\sum_{t=1}^T \left(\mathbb{E}[\tilde{\Psi}_{j,t}(u_{i,t})]- \tilde{\Psi}_{j,t}(u_t^*)\right)\notag\\
        =& \mathcal{O}( T^{\frac{9}{10}}(1+
        \sqrt{D_T})+ T^{7/10}).
    \end{align}
\end{theorem}
\begin{proof}
    Combining \eqref{Psi_ineq}, \eqref{Psi_dif_2} and \eqref{inn_prod_sum} yields
    \begin{align}\label{sum_ineq_1}
        &\sum_{j=1}^n\sum_{t=1}^T \left(\tilde{\Psi}_{j,t}(u_{i,t})- \tilde{\Psi}_{j,t}(u_t^*)\right)\notag\\
        \leq& \sum_{j=1}^n\sum_{t=1}^T \left(\tilde{\Psi}_{j,t,\delta}(w_{i,t})-\tilde{\Psi}_{j,t,\delta}(w_t^*)\right)+3 n \delta M T\notag\\
        \leq & \sum_{j=1}^n\sum_{t=1}^T \nabla \tilde{\Psi}_{j,t,\delta}(w_{j,t})^\top(w_{j,t}-w_t^*)+3 n \delta M T \notag\\
        &+2n M(\frac{8}{3}D T^{3/4}+\eta \alpha' \beta  T)\notag\\
        \leq & \sum_{j=1}^n \sum_{t=1}^T  \mathcal{P}_t+n\sum_{t=1}^T\bar{a}_t^\top (\bar{w}_t-w_t^*)\notag\\
        &+(n\beta+2n M)(\frac{8}{3}D T^{3/4}+\eta \alpha' \beta  T)+3 n \delta M T \notag\\
        \leq & \sum_{j=1}^n D\sum_{t=1}^T \|\nabla \tilde{\Psi}_{j,t,\delta}(w_{j,t})-a_{j,t}\| +n\sum_{t=1}^T\bar{a}_t^\top (\bar{w}_t-w_t^*)\notag\\
        &+(n\beta+2n M)(\frac{8}{3}D T^{3/4}+\eta \alpha' \beta  T)+3 n \delta M T.
    \end{align}

    Taking expectation of \eqref{sum_ineq_1} and using Proposition \ref{sum_norm_sub_lem}, we have 
    \begin{align}
        &\sum_{j=1}^n\sum_{t=1}^T \left(\mathbb{E}[\tilde{\Psi}_{j,t}(u_{i,t})]- \tilde{\Psi}_{j,t}(u_t^*)\right)\notag\\
     \leq & nD\left(\mathcal{O}(T^{1+b-\frac{c}{2}})+\mathcal{O}(T^{\frac{1}{2}+b+\frac{c}{2}})+\mathcal{O}(T^{c+\frac{1}{2}})\sqrt{D_T}\right)\notag\\
     &+n\sum_{t=1}^T\bar{a}_t^\top (\bar{w}_t-w_t^*)+3 n \delta M T\notag\\
     &+(n\beta+2n M)(\frac{8}{3}D T^{3/4}+\eta \alpha' \beta  T).
    \end{align}
    Then, further using Proposition \ref{at_in_bound_lem}, $\eta=T^{-3/4}$, $\delta=T^{-b}$, and $\rho=T^{-c}$, we obtain
    \begin{align}\label{sum_fin_ineq}
        &\sum_{j=1}^n\sum_{t=1}^T \left(\mathbb{E}[\tilde{\Psi}_{j,t}(u_{i,t})]- \tilde{\Psi}_{j,t}(u_t^*)\right)\notag\\
     \leq & nD\left(\mathcal{O}(T^{1+b-\frac{c}{2}})+\mathcal{O}(T^{\frac{1}{2}+b+\frac{c}{2}})+\mathcal{O}(T^{c+\frac{1}{2}})\sqrt{D_T}\right)\notag\\
     &+n\left(\frac{D^2}{\eta}+\eta T \beta^2\right)+3 n \delta M T\notag\\
     &+(n\beta+2n M)(\frac{8}{3}D T^{3/4}+\eta \alpha' \beta  T)\notag\\
     =&nD\left(\mathcal{O}(T^{1+b-\frac{c}{2}})+\mathcal{O}(T^{\frac{1}{2}+b+\frac{c}{2}})+\mathcal{O}(T^{c+\frac{1}{2}})\sqrt{D_T}\right)\notag\\
     &+n D^2 T^{3/4}+n\beta^2 T^{1/4}+3 n M T^{1-b}\notag\\
     &+(n\beta+2n M)(\frac{8}{3}D T^{3/4}+\alpha'\beta T^{1/4}).
    \end{align}
    To guarantee that the dynamic regret grows sublinearly with respect to the time horizon $T$ (up to the non-stationarity measure $D_T$), we set $b=1/10$ and $c=2/5$,
    \begin{align}
        &\sum_{j=1}^n\sum_{t=1}^T \left(\mathbb{E}[\tilde{\Psi}_{j,t}(u_{i,t})]- \tilde{\Psi}_{j,t}(u_t^*)\right)\notag\\
        =& \mathcal{O}( T^{\frac{9}{10}}(1+
        \sqrt{D_T})+ T^{7/10}).
    \end{align}
    This establishes a sublinear dynamic regret bound with explicit dependence on the temporal variation metric $D_T$. This completes the proof.
\end{proof}
\begin{remark}
The stepsizes in Theorem \ref{reg_them} are chosen in terms of the time horizon $T$ for notational simplicity. In practice, a horizon-free implementation can be obtained using the standard doubling trick technique with periodic restarts \cite{ShalevShwartz2011}. Under these modifications, the same order of regret bounds up to constant factors can be achieved.
\end{remark}

\section{Numerical simulation}\label{simulation}
In this section, we evaluate the performance of the proposed feedback DPFOA and compare it with existing works. Consider an undirected and connected multi-agent network with ten agents, i.e. $n=10$. The dynamics of agent $i$ are described by
\begin{align}\label{dynam_simu}
    x_{i,t+1}=&A_i x_{i,t}+B_i u_{i,t}+d_{i,x}\notag\\
    y_{i,t}=&C_i x_{i,t}+d_{i,y}.
\end{align}
where the disturbances $d_{i,x}$ and $d_{i,y}$ are independently generated from standard normal distributions. 
For simplicity, we consider three type of agents with different $A_i$, $B_i$ and $C_i$. Each agent is associated with a time-varying private loss functions $f_{i,t}$ and all agents cooperate to solve the following distributed online optimization problem:
\begin{align}\label{simu_pro}
    \min_{u,y} f_t(u,y)=&\sum_{i=1}^n f_{i,t}(u_i,y_i)\notag\\
    f_{i,t}(u_{i},y_{i})=&a_{i,t}\|u_{i}\|^2+b_{i,t}\|y_{i}-c_{t}\|^2,\notag\\
    \text{s.t.} \qquad y=&h(u), \ u_i \in \mathcal{U},
\end{align}
where $u_i$ and $y_i$ are the input and output of the dynamic system of agent $i$ in \eqref{dynam_simu}. The coefficients $a_{i,t}$ and $b_{i,t}$ are independently sampled from a standard uniform distribution. The mapping $h(u)$ is the steady-state input-output relationship of the system \eqref{dynam_simu}, and the input constraint set is set as $\mathcal{U}=\{u|\|u\|_1\leq 2\}$. The objective of all agents is to cooperatively drive the system output to track a moving target, whose position $c_t$ evolves according to
$$c_{t+1}=c_{t}+
\frac{\sin \frac{t}{100}}{10 t} ,$$
with $c_{0}=0$. At each time step $t$, every agent $i$ generates an input $u_{i,t}$ and observes its local cost function value $f_{i,t}$. All agents are initialized with the same feasible input $u_{i,0} \in \mathcal{U}$.
\begin{figure}
    \centering
    \includegraphics[width=7cm]{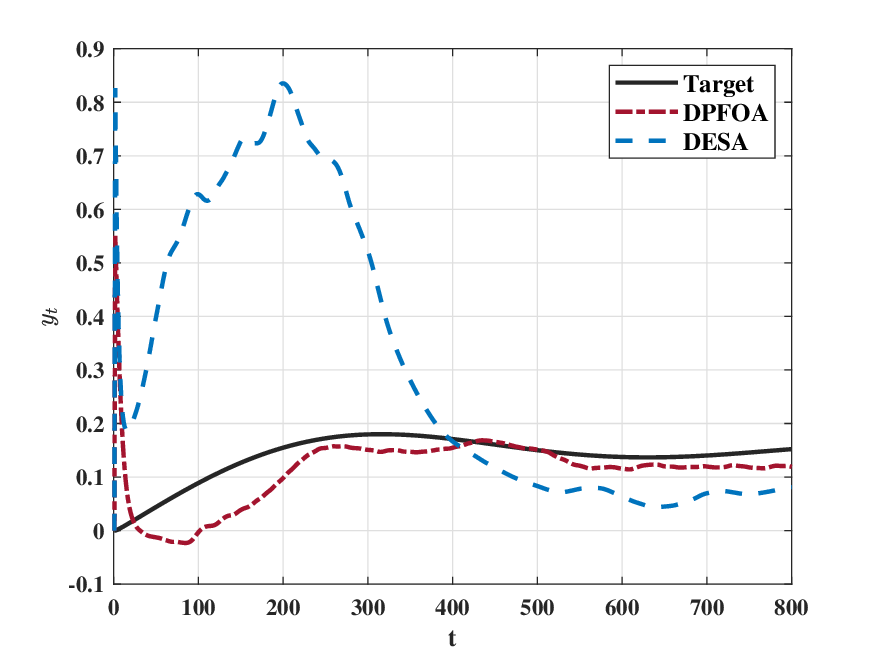}
    \caption{The paths of $c_{t}$ and $\bar{y}_t$ generated by DPFOA and DESA.}
    \label{traj_fig}
\end{figure}
\begin{figure}
    \centering
    \includegraphics[width=7cm]{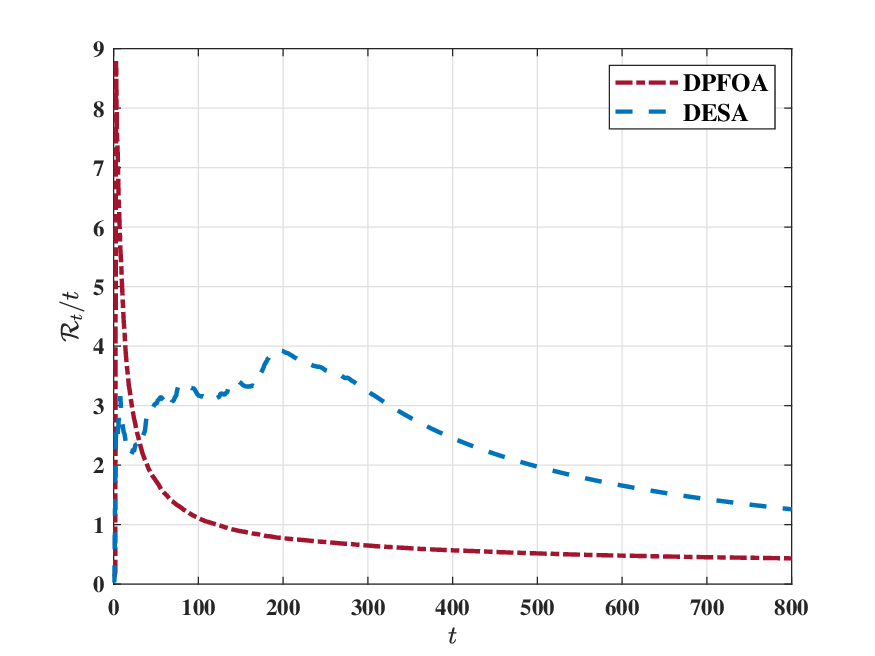}
    \caption{The regret trajectory of  DPFOA and DESA.}
    \label{reg_fig}
\end{figure}
\par For performance comparison, we apply both the proposed DPFOA and a distributed version of the stochastic extremum-seeking algorithm (DESA) adapted from \cite{ES_alg} to problem \eqref{simu_pro}. Since the ES algorithm in \cite{ES_alg} was originally designed for a centralized setting, we extend it to the distributed framework by incorporating an average-consensus step. Both algorithms operate as model-free feedback frameworks. Fig. \ref{traj_fig} illustrates the target trajectory $c_{t}$ along with the trajectories of average output $\bar{y}_t=\frac{1}{n}\sum_{i=1}^n y_{i,t}$ for both algorithms. It is observed that both algorithms are capable of tracking the moving target asymptotically. However, the trajectory generated by DPFOA converges faster and remains closer to the true trajectory throughout the tracking process, while DESA exhibits larger transient deviations. 
\par To further evaluate the regret performance, we define the worst-case regret across all agents as $\mathcal{R}_t=\max_{i\in [10]}\mathcal{R}_{i,t}$, and present its time-averaged value $\mathcal{R}_t/t$ in Fig. \ref{reg_fig}. Both algorithms exhibit decreasing regret curves, confirming their sublinear regret behavior. Nevertheless, DPFOA achieves substantially faster decay and attains a much lower regret level than DESA, demonstrating its superior efficacy and tracking accuracy.

\section{Conclusion}\label{conclusion}
This paper has developed a distributed projection-free bandit online feedback optimization algorithm for multi-agent dynamical systems with constrained inputs and time-varying cost functions. By utilizing a smoothing zeroth-order one-point estimator, the proposed algorithm constructs local gradient approximations directly from real-time input–output data, eliminating the need for any system model. A projection-free conditional gradient update, integrated with distributed global gradient estimators, is employed to effectively handle input constraints and maintain scalability in large-scale networked systems. The distributed online algorithm is shown to achieve a sublinear dynamic regret. Numerical simulations illustrate its effectiveness. Future work will explore communication-efficient implementations and extensions that incorporate additional safety or state constraints.

\bibliographystyle{IEEEtran}
\bibliography{refer}
\end{document}